\documentclass[11pt]{scrartcl}  
\usepackage[english]{babel}
\usepackage[latin1]{inputenc}
\usepackage[T1]{fontenc}
\usepackage{amsmath}
\usepackage{amsfonts}
\usepackage{amssymb}
\usepackage{amsthm}
\usepackage{appendix}
\usepackage{enumitem}
\usepackage{bbm}
\usepackage{xcolor,graphicx}
\usepackage[affil-it]{authblk}
\usepackage[left=3.5cm, right=3.5cm, bottom=4.3cm]{geometry}
\newtheorem{theorem}{Theorem}[section]
\newtheorem{lemma}[theorem]{Lemma}

\theoremstyle{definition}
\newtheorem{definition}[theorem]{Definition}

\theoremstyle{remark}

\numberwithin{equation}{section}

\newcommand{\be}{\begin{equation}}
\newcommand{\ee}{\end{equation}}
\newcommand{\ba}{\begin{eqnarray}}
\newcommand{\ea}{\end{eqnarray}}

\newcommand{\R}{\mathbb{R}}

\def\G{\mathbb G}
\def\C{\mathbb C}
\def\deg {\rm {deg}}
\def\eqref#1{(\ref{#1})}

\usepackage{dsfont}
\def\1{ \mathds{1} }

\DeclareOldFontCommand{\rm}{\normalfont\rmfamily}{\mathrm}
\DeclareOldFontCommand{\sf}{\normalfont\sffamily}{\mathsf}
\DeclareOldFontCommand{\tt}{\normalfont\ttfamily}{\mathtt}
\DeclareOldFontCommand{\bf}{\normalfont\bfseries}{\mathbf}
\DeclareOldFontCommand{\it}{\normalfont\itshape}{\mathit}
\DeclareOldFontCommand{\sl}{\normalfont\slshape}{\@nomath\sl}
\DeclareOldFontCommand{\sc}{\normalfont\scshape}{\@nomath\sc}

\begin{document}

\title{ Edge switching transformations of quantum graphs}
\author{M. Aizenman, H. Schanz, U. Smilansky, S. Warzel}
\affil{\small
MA:  Departments of Physics and Mathematics, Princeton University,
Princeton NJ 08540, USA\\ 
HS: Institute of Mechanical Engineering, University of Applied Sciences
Magdeburg-Stendal, D-39114 Magdeburg, Germany. \\
US:  Department of Physics of Complex Systems,
Weizmann Institute of Science, Rehovot 7610001, Israel\\ 

SW: Zentrum Mathematik, TU M\"unchen, 
Boltzmannstr. 3, 85747 Garching, Germany. }
%
\date{\small \today}
\maketitle
\begin{abstract}
Discussed here are  the  effects of basics graph transformations on the spectra of associated quantum graphs.  
In particular it is shown that under an edge switch the  spectrum of the transformed Schr\"odinger operator 
 is interlaced with that of the original one. 
  By implication, under edge swap the spectra before and after the transformation,  
 denoted by $\{ E_n\}_{n=1}^{\infty}$ and $\{\widetilde E_n\}_{n=1}^{\infty}$ correspondingly, are level-2 interlaced, so that $E_{n-2}\le \widetilde E_n\le E_{n+2}$.   The proofs are guided by considerations of  the quantum graphs' discrete analogs. 
\end{abstract}

\section{Statement of the main result}

Quantum graphs are linear, network-shaped structures of vertices connected by edges, with a Schr\"odinger like  operator suitably defined on functions supported on the edges.   Such systems were  studied by Linus Pauling as simplified models of valance electrons in organic molecules in the 1930s~\cite{Paul}.  
More recently attention was called to quantum graphs in the context of quantum chaos  \cite{kottos1,kottos2},   
and they are defined and discussed at great length and detail in review articles and books (see e.g., \cite{sven,kuch_08,graphs,Post}).   
In this note, we consider a general compact quantum graph,
$\G = (\mathcal V, \mathcal E,\mathcal L)$,   whose edges $e\in \mathcal E$ are metrized and of finite lengths belonging to the list  $ \mathcal L = \{L_e\}_{e=1}^{|\mathcal {E}|}$.
The associated Schr\"{o}dinger operators $H$ act in $L^2(\G):= \bigoplus_{e\in \mathcal E} L^2( e, dx)$, may include external potential $V$ and possibly also a magnetic potential $A$:
\be \label{eq:H}
H \ = \   \left (\frac 1 i  \frac d {dx} -A(x)  \right)^2 \ + \  V(x) \,.
\ee
The operator's definition is incomplete without specifying also the boundary conditions~(bc) at vertices.  They are assumed here to be local, i.e., expressed in terms of linear relations on the   limiting values of the function and of its derivatives along the ${{\deg}} (v)$ edges which meet at each vertex $v\in \mathcal V$.  The possible choices which ensure  self-adjointness  are reviewed in detail  in, e.g.,~\cite{kuch_08,graphs}.\\ 

Under mild conditions on $V$ and $A$ the spectrum of the Schr\"{o}dinger operator $\{ E_n\}_{n=1}^{\infty}$ is discrete and bounded below, but not above.   
It suffices to assume $V , A  $ are integrable, but for a more transparent presentation we focus on the case these are piecewise continuous~\cite{graphs}.  \\

This spectrum of $H$ is unaffected by the operation of \emph{edge splitting},   through the insertion of a  vertex of degree $2$ with Kirchhoff boundary conditions.  These require  $ \Psi $ to be  continuous at $ v \in \mathcal{V} $ having there directional derivatives satisfying: 
\be\label{eq:Kirchoff}
 \sum_{e \in \mathcal{E}_v}  \frac{d}{dx_e} \Psi(v) \  = \ 0 . 
 \ee 
(To avoid confusion:  these  boundary conditions are not assumed here for the other quantum graph vertices.)\\

Our purpose here is to discuss the effects on the spectrum of another basic graph transformation, that of  \emph{edge switching}.  It is  defined in the following extension of  a notion which is used in combinatorial graph theory \cite{Seidel}.\\ 

\begin{definition}  \label{def:switch}
For a quantum graph $\G$ with a self-adjoint Schr\"{o}dinger operator $H$, 
an  \emph{edge switch}  is a transformation in which 
a pair of edges  in $\mathcal E $ exchange the graph designations of one of their end points.  Preserved under this exchange are  the  edge lengths,  the local action of $H$ along the corresponding edges, and  the vertex boundary conditions --  up to the corresponding transposition of the functions whose limiting behavior at the two vertices enter the local boundary conditions.  
\end{definition}

Under an edge switch the  collection of the edge lengths remains unchanged. However the spectra will in the generic case be affected. 
For example, for quantum graph Laplacians it is shown in  \cite{gutkin}   that the topology of the quantum graph with rationally independent edge length  $\mathcal {L}$  "can be heard" in the sense of Marc Kac \cite {marckac}.  For clarity let us add that the graph topology may be affected by a switch, but it does not have to.  However, regardless of that, and even in case the switch involves two end points which terminate at a common vertex, the spectrum may be affected  through the resulting transposition in the   boundary conditions.   \\ 

Edge switch may  be combined with the actions of edge splitting and its inverse -- the removal of a degree $2$ vertex of Kirchhoff boundary conditions.   The group generated by these elements includes various other useful 
 graph modifications under which only the total length $\sum_{e\in \mathcal E} L_e$ is preserved.   Among those are: 
\begin{enumerate}[label=\alph*.]
\item An  \emph{edge crossing}:  a transformation in which a pair of edges $ e , e' $ are cut at 
 points $ s_{e} \in (0, L_e) $, $ s_{e'}  \in (0, L_{e'}) $,  and cross-rewired at these cut points, keeping the action of $ H $  locally unchanged along the edge segments.  The map  is done using vertex insertion, with Kirchhoff rules, followed by an edge switch at the inserted vertices, and then  removal of the added pair. 
\item  An  \emph{egde reversal}:  a mapping in which the parametrization of one of the graph's edges is reversed, i.e.,  the potentials $V$ and $A$ on $ e $ are replaced by
$V(L_e- \cdot )$ and $ A(L_e- \cdot )$, with no other change, in particular preserving the graph's structure, its edge lengths, and the collection of the boundary conditions  through which $H$ is defined.
\item An \emph{edge swap}:  a transformation in which  the lengths of the two edges and the action of $H$ along them are exchanged while preserving the graph structure and the boundary conditions at all vertices.
 \item An \emph{edge  segment exchange}: in which a pair of edges exchange mid segments.  The mapping is realizable through a pair of edge crossings.    
\end{enumerate} 
It may be noted that the graph topology, which is not preserved under a generic edge switch, is preserved under any of the last three  operations.  \\

The main result reported here is a  statement constraining the difference between the quantum graph's spectra before and after any of the above operations.  We use the following terminology.

\begin{definition} \label{def:intertwine}
Two  non-decreasing sequences $\{E_n\}$ and  $\{\widetilde E_n\}$ are said to be \emph{degree-$r$ interlaced}  if  for any $n>r$
\be \label{inter}
E_{n-r}\le \widetilde E_n \le E_{n+r} \,.
\ee
In case $r=1$ the sequences are more simply said to be interlaced. 
\end{definition}
The above is a symmetric relation.  Eq.~\eqref{inter}  may be equivalently  expressed  by saying that the spectral shift is uniformly bound by $r$, i.e.,  for all $E\in \R$ 
\be 
|\widetilde N(E) -N(E) | \, \leq  \, r \,    
\ee 
where $N$ and $\widetilde N$ are the  counting functions 
\be
N(E) \ = \  \rm{card}\{n \,| \, E_n\leq E \}\, , \quad \widetilde N(E) \ = \  \rm{card}\{n \,| \, \widetilde E_n\leq E\} \,.
\ee 

Following is our main result.

\begin{theorem}  \label {thm:r}  
Let  $(\G,H)$  be a quantum graph, which need not be connected, with a self-adjoint Schr\"{o}dinger operator of the form \eqref{eq:H}.  Then under any edge switch the spectra before and after the transformation are interlaced (at degree $r=1$).  
\end{theorem}

An immediate consequence is that under each of the other transformations mentioned above the spectra before and after the transformation  are degree-$r$ interlaced, with
\be \label{eq:r}
r \ = \ \begin{cases} 
         1 & \mbox{under a single edge crossing}  \\ 
         1 & \mbox{edge reversal}  \\ 
         2 & \mbox{edge  swap}  \\ 
         2 & \mbox{partial exchange of edge segments}  \\ 
                     \end{cases} 
\ee 
Further implications, and related research questions,  are discussed at the end of the article.\\ 

\begin{figure}[ht]
\centerline{\includegraphics[width= .3\textwidth, angle = -90]{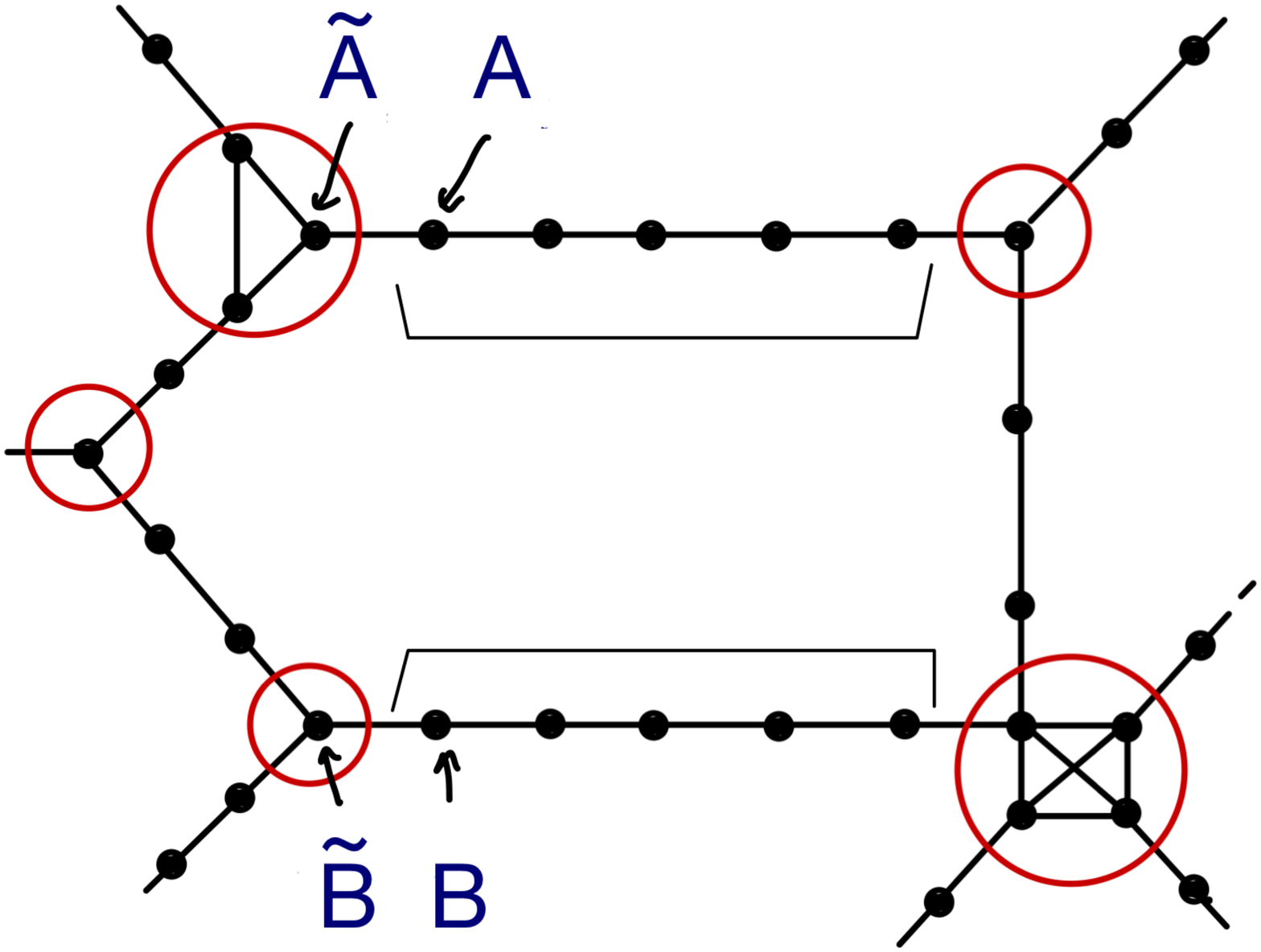}}
\caption{Piece of a discrete version of a quantum graph.  The quantum graph's vertices, marked here by circles, correspond to connected clusters of discrete vertices of degrees other than 2.  Quantum graph's edges correspond to chains of discrete vertices of degree 2.}
\label{fig:discrete}
\end{figure}

The notions discussed above can also be considered in the context of finite discrete graphs, an  example of which is depicted in Fig.~\ref{fig:discrete}.  The discrete Hamiltonian $H$ is given by a hermitian matrix whose off diagonal terms are of the form
\be 
H_{u,v}   \ = \   J_{u,v} e^{i \theta(u,v)}
\ee  
with $J_{u,v}$ the adjacency matrix of the graph plus an additional potential $ V $ on the diagonal $ u = v $ and $\theta(u,v) = - \theta(v,u)$ representing the single step integral of the vector potential $A$.  The role which for quantum graphs is played by edges is assumed here  by chains of vertices of $\deg (v)=2$.   The end points of any chain are its extremal points of degree 2.  
The discrete analog of the quantum graph's vertices are therefore the  connected clusters of points with $\deg(v) \neq 2$.    

Using  the above as a dictionary, the different notions of graph transformations have a natural extension to the discrete graph case, and so does Theorem~\ref{thm:r}.   

\section{A pair of relevant finite-rank perturbation principles} 

It is somewhat instructive to consider first the discrete version, which is what we shall do.  
As a preliminary observation, let us note that for any edge switch $S=S^*$  the difference $H- SHS$ is an operator of rank at most $4$.    E.g.,  the switch  $S_{A,B}$  of the two edges highlighted in Fig.~\ref{fig:discrete}, affect only the matrix elements of $H$ within the $4$ dimensional block spanned by the vectors 
$ \{ | A \rangle, \, | \widetilde A \rangle,\,  | B \rangle, \, | \widetilde  B\rangle\}$.   

By the following finite-rank perturbations principle   (cf.~\cite{S_rank1})
any low rank perturbation has only a limited effect  on the spectral counting function 
\be 
N(E;H) \ := \  \rm{dim \, Range}\, \, P_{H< E}  \, .
\ee  
 
\begin{lemma}  
Let $(H,K)$ be a pair of self adjoint operators, with $H$ 
of discrete spectrum which is bounded below  
 and  $K $ of a finite rank.   Then  for any $ E \in \mathbb{R} $: 
\be
  | N(E; H) - N(E;H+K)|  \leq   \rm {Rank} \, K    \, . 
\ee   
\end{lemma}
The quantity $\xi(E;H,H') := N(E; H)-N(E;H')  $ is often referred to as the Krein spectral shift. 

This already allows to deduce (initially for the discrete case) 
a weaker version of Theorem~\ref{thm:r}, with \eqref{eq:r} replaced by  $4$ times  as large spectral shift bounds.\\ 

One can do a bit better using the monotonicity of the spectral function under positive perturbations, and the fact that $H$ and $S_{A,B}H S_{A,B}$ agree on the subspace of functions which vanish at $A$ and $B$. 
For that, one may consider the two families of operators
\begin{eqnarray}  
H_\lambda \,  :=&\  H\qquad  +  & \lambda \ | A\rangle \langle A | + \lambda \ | B\rangle \langle B |  \nonumber \\       
\widetilde H_\lambda  \, :=&   S_{A,B}H S_{A,B} \,   + & \lambda \ | A\rangle \langle \  A | + \lambda \ | B\rangle \langle B | 
\end{eqnarray}  
 
 The key observation now is that as $\lambda$ increases from $0$ to $+\infty$ the spectral counting functions $N(E,H_\lambda)$ and $N(E,\widetilde H_\lambda)$  can only decrease, but by not more than $2$,     and for each $E<\infty$    
 \be
 \lim_{\lambda \to \infty} N(E,H_\lambda) \ = \ \lim_{\lambda \to \infty}  N(E,\widetilde H_\lambda) 
 \ee 
This allows to deduce that
\be 
| N(E,H) - N(E,S_{A,B}H S_{A,B} ) |\, \equiv \,  | N(E,H_0) - N(E,\widetilde H_0 ) | \ \leq \, 2 \,,
\ee  
 which is a step closer to our claim.  \\

For  the tighter bound which is stated in Theorem~\ref{thm:r} we shall  use  the following comparison  principle. 

\begin{lemma} \label{lem:reflection} 
For any pair of self adjoint operators $H_0$ and $K$, with  $K$ of finite rank, and $T$ a unitary  matrix  such that $T^2=\1$ (equivalently $T^*=T$) the spectra of $H+K$ and $H+TKT$ are intertwined to degree 
$ \rm{rank}(K) /2 $, i.e. for all $E \in \R$
\be 
|\xi(E; H_0 +TKT, H_0+K)| \ \leq \ \tfrac 12 \  \rm{Rank} (K)   \, . 
\ee 
\end{lemma}

\begin{proof}
The relation between the two operators can be rewritten as 
\be 
H_0 + TKT \ = \    (H_0 + K)  +  \Delta H\, , 
\ee 
with $\Delta H := T^* K T- K  $.   
The key observation now is that 
\be 
 T^* \Delta H T \ = \  -  \Delta H  \, .
 \ee 
Hence, in the spectral decomposition 
\be 
 \Delta H \ = \  \sum_{j=1}^{\rm{rank}(K)}  \lambda_j | \Psi_j\rangle \langle \Psi_j| 
\ee 
the number of strictly positive eigenvalues is at most $\rm{rank}(K)/2$, and so is the number of strictly negative eigenvalues.   (In the present case the two are equal, due to the antisymmetry, but it is the above property which matters.)

Therefore adding $ \Delta H  $ is equivalent to the addition, one at a time, of $\rm{rank}(K)$ operators of rank $1$.   In this process, for each $E$  the spectral counting function $N(E)$ changes by at most $\rm{rank}(K)/2$ steps of $+1$  and at most  $\rm{rank}(B)/2$ steps of $-1$.   Thus the net spectral shift is bounded by  $\rm{rank}(K)/2$. 
\end{proof}

\section{The discrete graph case} 

The argument which will be employed for the optimal uniform spectral shift bound (see Fig.~\ref{fig3}) starts  from a split of the graph at the two points $A$ and $B$, in a manner indicated in Fig.~\ref{fig:split}.  
\begin{figure}[h]
\centerline{\includegraphics[width= .4\textwidth, angle = -90]{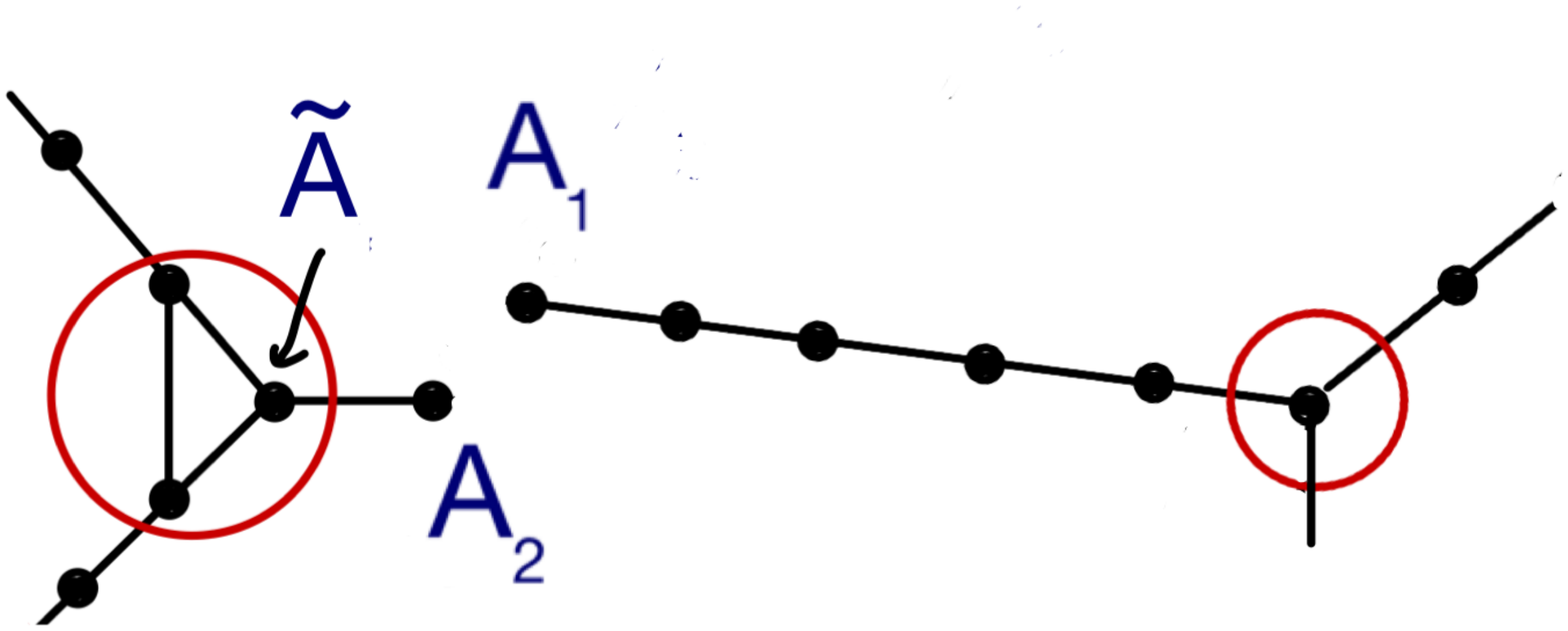}}
\caption{The vertex split $A \mapsto (A_1,A_2)$ which is employed in the proof of the discrete version of Theorem~\ref{thm:r}.  
A similar split is made at point $B$ of Fig.~\ref{fig:discrete}}
\label{fig:split}
\end{figure}
\begin{proof}[Proof of Theorem~\ref{thm:r} -- the discrete case]
Consider the discrete analogue of the quantum graph which is obtained by splitting the end points of the edges which are to be switched, in a manner indicated in Fig.~\ref{fig:split}.   Denoting by $\mathcal H $ the original finite-dimensional Hilbert space and by  $\widehat {\mathcal H}$ the one  corresponding to the enlarged graph. 
The latter is isomorphic to the direct sum $\mathcal H \bigoplus  \C^2  $.  
To the  operator $H$ on  $\mathcal H $, we associate  the  operator $\widehat H$  on  $\widehat {\mathcal H}$ through the following choices: i)  The two operators coincide within the subspace which does not involve the end points $A$ and $B$, ii) The matrix element of $ \widehat H $ between  $|A_1\rangle $ and its neighbor along the edge equals $ \sqrt{2}  $ times that of $ H $ between $A $ and that neighbor.  iii) The matrix element of  $ \widehat H $ between $| A_2\rangle $ and the vertex site
is again the product of $ \sqrt{2} $ and the corresponding matrix element of $ \widehat H $. And finally iv)  $V_A := \langle A_1|  \widehat H |A_1\rangle  =  \langle A_2|  \widehat H |A_2\rangle  =  \langle A|  H |A \rangle $.
Similar convention will be applied at   the endpoint $| B \rangle $ with which $| A \rangle $ is to be switched. 
Consider now the one parameter family of operators 
 \be 
 H_\lambda \ =\  \widehat H + \lambda \left[  
    (  | A_1\rangle -| A_2\rangle) ( \langle A_1 |  - \langle A_2 |)  + 
       (  | B_1\rangle -| B_2\rangle) ( \langle B_1 |  - \langle B_2 |) 
 \right]  \,.  
 \ee 
Denoting by  $S=S^*$ the operator which transposes  $|A_2\rangle  \leftrightarrow | B_2\rangle $, we find: 
\be 
 S^*H_\lambda S  -  H_\lambda  \ = \  \lambda \ (  | A_1\rangle -| B_1 \rangle) ( \langle A_2 |  - \langle B_2 |)  + 
      \lambda \ (  | A_2\rangle -| B_2\rangle) ( \langle A_1 |  - \langle B_1 |)
\ee     
This difference is rank $ 2 $ and antisymmetric with respect to the transposition $ S $. Lemma~\ref{lem:reflection} thus implies that 
\be
|\xi(E; H_\lambda,  S^*H_\lambda S )| = \left| N(E; H_\lambda) - N(E; S^*H_\lambda S) \right| \leq 1 \, . 
\ee
By an orthogonal change of basis in the subspace $ \{ | A_1\rangle , | A_2\rangle , | B_1\rangle , | B_2\rangle \} $ to 
$ \{ \frac{1}{\sqrt{2}} (  | A_1\rangle + | A_2\rangle) ,  \frac{1}{\sqrt{2}}  (  | A_1\rangle -| A_2\rangle) ,  \frac{1}{\sqrt{2}} (  | B_1\rangle + | B_2\rangle) ,  \frac{1}{\sqrt{2}}  (  | B_1\rangle -| B_2\rangle) $, one realizes that $ H_\lambda $ on $\widehat {\mathcal H} $ is unitarily equivalent to  $ H \oplus {\rm diag}(2\lambda+V_A,2\lambda+V_B) + X $ on a space equivalent to $ {\mathcal H} \oplus \C^2 $ with $ X $ 
a $\lambda$- independent rank-$4$ operator which  is purely off diagonal in the above $\rm{dim} \mathcal {H} \times 2$ block decomposition of $\widehat {\mathcal H} $.  I.e., 
denoting by $ Q $ the orthogonal projection onto the $ \C^2 $-component of $\widehat {\mathcal H} $ and by $ P = \1 - Q $, one has $ PX P = Q X Q = 0 $. Consequently through an application of the Schur complement formula, for any $ E \in \mathbb{R} $:
\be
\lim_{\lambda\to\infty} N(E;H_\lambda) = N(E;H)
\ee
and the same applies to $ S^*H_\lambda S $. Hence, the spectral shift under edge switch is bounded by one as claim. 
\end{proof}

\section{Proof of the main result}

While Theorem~\ref{thm:r} addresses the spectral shift caused by an edge switch, it is convenient to regard this transformation  as a limiting case of an edge crossing.   More specifically, let $ e \equiv [0,L_e] $ and $ e'  \equiv [0,L_{e'}] $ be a pair of edges whose end points are to be switched. Splitting the edges through the insertion of vertices  located at a distance $ \varepsilon > 0 $ from the corresponding endpoints, with the  Kirchhoff boundary conditions on the inserted vertices leaves the spectrum of Hamiltonian $ H $ invariant. Crossing the edges at the split points results in a new Hamiltonian which we denote by $ H_\varepsilon $. As $ \varepsilon \downarrow 0 $ this Hamiltonian converges in norm resolvent sense to the Hamiltonian of the edge switch dealt with in Theorem~\ref{thm:r}. 
This implies the convergence of the counting functions $ N(E;H_\varepsilon) $ for almost all $ E \in \mathbb{R} $. It therefore suffices to focus in the proof on the spectral shift under an edge crossing transformation.

\begin{proof}[Proof of Theorem~\ref{thm:r} -- continuum  case, for  edge crossing] \mbox{} \\   
With Kirchhoff  boundary conditions at the above described edge insertion sites on the chosen pair of edges,  the insertion in essense leaves the Hamiltonian invariant.  More precisely, it causes no change in the spectrum.  However, changing the boundary conditions there to Dirichlet  one obtains a different  Hamiltonian which we denote by $ H_0 $. 
As a tool in the proof we will work with the spectral shift resulting from this change: 
\be
\xi(E;H,H_0) = N(E;H) - N(E;H_0) 
\ee
where $ N(H;E) $ stands for the number of eigenvalues of $ H $ strictly below $ E $.

Let $ \widetilde H $ stand for the Hamiltonian in which the two edges are crossed at the two inserted vertices. By the additivity of the spectral shift, the quantity of interest can therefore be written as
\be
\xi(E;H,\widetilde H) \ = \  \xi(E;H,H_0) - \xi(E;\widetilde H,H_0) \, . 
\ee
Since we assumed the quantum graph to be compact, the spectral shift functions are piecewise constant with only countably many discontinuities. It therefore remains to bound the spectral shift only for almost every $ E \in \mathbb{R} $. 

The resolvents of the original Hamiltonian $ H $ and the Dirichlet-decoupled operator $ H_0 $ are related by 
the Krein formula 
\be
(H-z)^{-1} -  (H_0-z)^{-1} =  - \gamma(z) M(z)^{-1} \gamma(z^*)^*
\ee
which involves the gamma field $\gamma $ and the Weyl function $ M $ corresponding to this change of boundary conditions~\cite{AlbPush05}. The latter is a $ 4\times 4 $ matrix-valued Herglotz function $ z \mapsto M(z) $ of the spectral parameter $ z \in \mathbb{C}^+ $. The Weyl function's boundary values $ M(E) := \lim_{\varepsilon \downarrow 0} M(E+i\varepsilon) $ exist for almost all $ E \in \mathbb{R} $.  Since the graph is compact, these boundary values are self-adjoint $ M(E)^* = M(E) $, 
We now recall the result of Behrndt-Malamud-Neidhardt~\cite[Thm.~4.1]{BMN08} that for almost all $ E \in \mathbb{R} $ the spectral shift is given by
\be
\xi(E;H,H_0)  \ = \  \frac{1}{\pi} \rm{Arg} \det M(E) = N(0; M(E)) 
\ee
where $ \rm{Arg}: \mathbb{C}\backslash \{0\} \to  (-\pi,\pi] $ is the imaginary part of the principle value of the complex logarithm.   Likewise, the spectral shift for the edge-crossed Hamiltonian $ \widetilde H $ is given in terms of its Weyl function $ \widetilde M(z) $. The latter is related to $ M(z) $ by a unitary transposition $S= S^* $ which swaps two rows/columns,
\be
 \widetilde M(z) = S \ M(z) S\, . 
\ee
Consequently, for almost every $ E \in \mathbb{R} $
\be 
\xi(E;H,\widetilde H) = N( 0; M(E))  - N(0; S M(E) S ) \, . 
\ee
Since the difference $  M(E) - S M(E) S $ is rank two and antisymmetric with respect to $ S $,  Lemma~\ref{lem:reflection} applies, and implies that the spectral shift is bounded by one. 
\end{proof}

\bigskip

\section{Discussion}

A few  comments are in order.\\

 -  A common feature of the transformations considered here is that they do not affect the asymptotic rate of growth of the spectrum.   This follows from  Weyl's semi-classical assymptotic formula \cite{Weyl11} (see also \cite{LT}):
\begin{eqnarray}
\lim_{E\to \infty}  \frac {1}{ \sqrt{E}}  N(H,E) &=&
  \lim_{E\to \infty}
\frac {1}{ \sqrt{E}}  \ \frac {1}{2\pi} \sum_{e\in \mathcal E(\G)}  \int_0^{L_e} dx  \int _{-\infty} ^{\infty} \1[ p^2 + V(x) \leq E]  \, dp \vspace{3cm}  \nonumber
\\ 
&  = &   \frac {1}{\pi} \sum_{e\in \mathcal E(\G)} |L_e| \end{eqnarray}
and the fact that the sum  of edge lengths is preserved under the transformations considered here. \\

-  Given a quantum graph $\G$ and a list $\mathcal L$ of $|\mathcal E|$ distinct lengths
there are  $|\mathcal E|\, !$ different arrangements of the lengths over the edges.
Denote  by $\mathcal{Q} = \mathcal{Q}_{\mathcal L }$ the set of the corresponding quantum graphs, at zero potential but  with some fixed boundary conditions.
This
provides an ensemble of quantum graphs with spectra of the same mean spectral density and a common structure of periodic orbits.   What would be the spectral statistics for this ensemble? How will it depend on the underlying common structure?\\

  -  For pairs of quantum graphs    $q, \tilde q  \in \mathcal{Q}$
we can define the (discrete) distance $\Delta(q,\tilde q )$  as the smallest number of elementary swaps by which  $q$ may be transformed to $\tilde q  $.
Thus we can define   an adjacency relation on $\mathcal{Q}$ between quantum graphs at distance $\Delta = 1 $.  This  in turn defines a d-regular "meta-graph" $G_{\mathcal{Q}}$ with the $|\mathcal E|\, !$ quantum graphs as vertices, each with the degree $d = |\mathcal E| (|\mathcal E|-1)/2$.   A random walk on $G_{\mathcal{Q}}$ corresponds to a succession of elementary swaps which are chosen at random with equal probability at each step. Since $G_{\mathcal{Q}}$ is d-regular, the random walk is ergodic and covers  $G_{\mathcal{Q}}$ uniformly after sufficiently long "time". This process is the discrete version of  Dyson's Brownian motion approach to spectral statistics  \cite {Dyson}, and it is discussed at length in \cite{Bernoulli}  for a different class of matrix ensembles. \\

- A related subject for future research is  the systematic study of correlations between different spectra and their dependence on the distance  $\Delta (q,\tilde q )$ between the graphs.\\

-  A possible tool for further insights on the spectral shift, and also an object of intrinsic interest  (cf.~\cite{BB}) is the question how are nodal counts  affected by edge swapping and the other transformations discussed here.\\

\begin{figure}[ht]
\centerline{
\includegraphics[width= .45\textwidth]{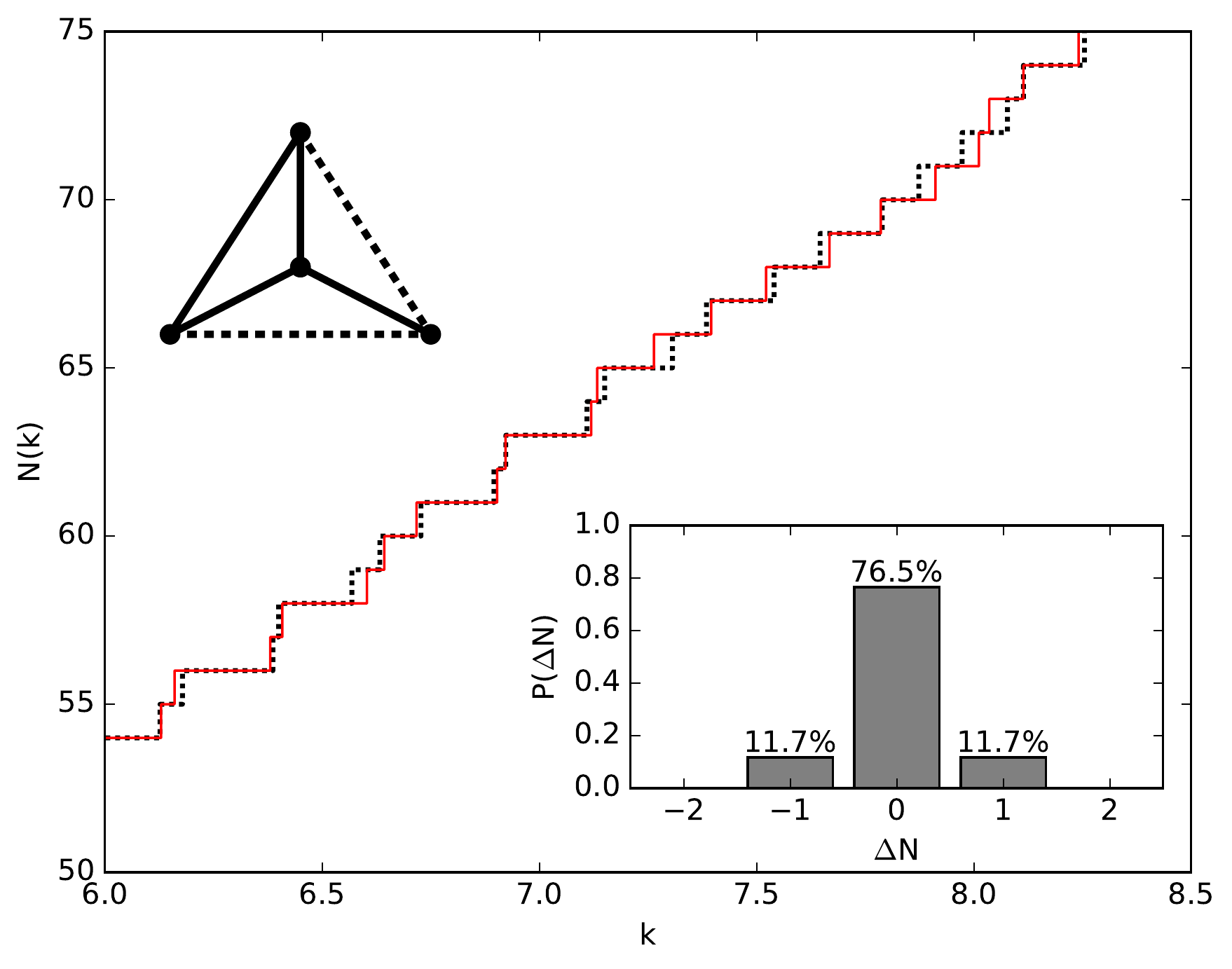} \quad 
\includegraphics[width= .45\textwidth]{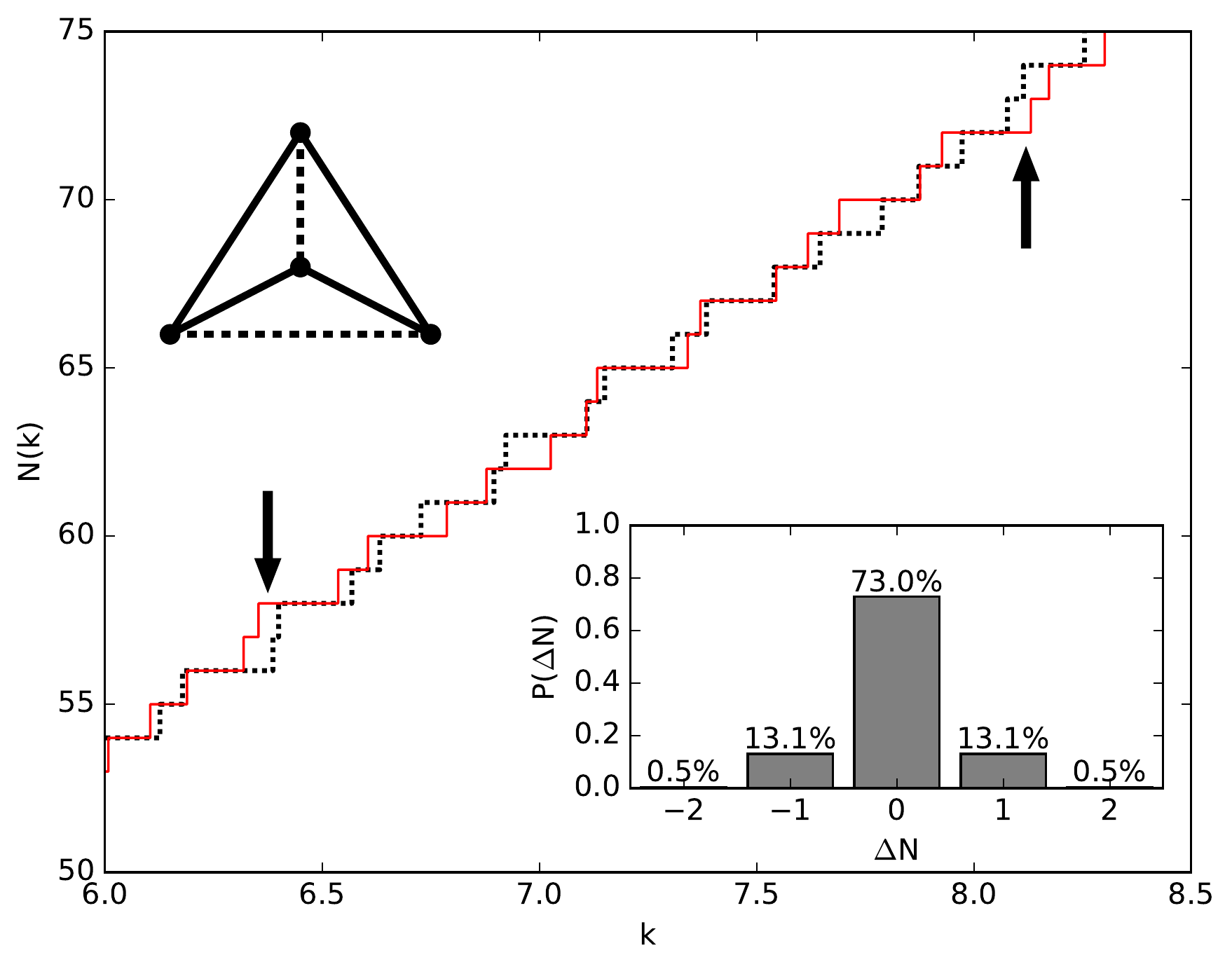}} 
\centerline{ i) \hspace{6
cm} ii)} 
\caption{Comparison of the spectral counting functions for the graph Laplacian on a
  tetrahedral graph of different edge lengths, under: i)  an edge switch and ii) edge swap.   
The insets show
  the corresponding numerical frequency distributions of the difference $\Delta N(E)$ at a randomly
  chosen energy (uniformly in $\sqrt{E}$),  computed from 10,000 spectral levels.  As is proven in this work, $|\Delta N|\le 1$ for any edge switch.  Under edge swaps our bound  $|\Delta N|\le 2$ $(r=2$ in
  \eqref{eq:r})  is attained, but not very frequently.} 
\label{fig3}
\end{figure}

- Length swapping can be easily implemented in experimental simulations of quantum graphs by networks of RF wave-guides \cite {sirko}.

\noindent {\small{\bf Note added in proof:}  After the submission of the paper it was called to our attention that an alternative protocol for the edge switch can be based on the vertex gluing operation, which was  discussed in \cite[Thm.~3.1.10]{graphs} and  in \cite[Thm.~1]{KMN13}.   Using it, an edge switch can be obtained through vertex gluing followed by un-gluing into the modified graph configuration.  Each step produces a spectral shift of at most one, but in different direction.   The spectral shift bounds presented here are not improved  by this comment, but it does offer another useful perspective.   We thank Gregory  Berkolaiko, Pavel Kurasov and  Sergey Naboko for alerting us to this point of view. } 

\section*{Acknowledgements}
We thank Chris Joyner, and Rami Band for  useful suggestions and corrections, and Mor Rozner  for  earlier numerical simulations which have guided us towards tighter results.
The work of MA was supported in part by the NSF grant DMS-1613296
and the Weston Visiting Professorship at the Weizmann Institute.  He thanks the Faculty of Mathematics and Computer Sciences and the Faculty of Physics  at WIS  for the hospitality enjoyed there.

\end{document}